\documentclass[11pt]{article}
\usepackage[small, bf]{caption}
\usepackage{indentfirst,url,amsfonts, amsthm, amsmath, amssymb,color, enumerate,algorithmic,verbatim,printlen}
\usepackage[boxed]{algorithm}
\usepackage[top=1.0in, bottom=1.0in, left=1.0in, right=1.0in]{geometry}
\usepackage{graphicx}

\newcommand{\ignore}[1]{}

\newtheorem{theorem}{Theorem}[section]
\newtheorem{lemma}[theorem]{Lemma}

\newtheorem{remark}[theorem]{Remark}

\newcommand{\poly}{\operatorname{poly}}

\newcommand{\E}{\operatorname{E}}

\newcommand{\defeq}{\stackrel{\textrm{def}}{=}}

\long\def\symbolfootnote[#1]#2{\begingroup%
\def\thefootnote{\fnsymbol{footnote}}\footnote[#1]{#2}\endgroup} 

\begin{document}
\begin{titlepage}
\title{A Distributed Minimum Cut Approximation Scheme}
\author{Hsin-Hao Su \footnote{This work is supported by NSF grants CCF-0746673, CCF-1217338, and CNS-1318294
and a grant from the US-Israel Binational Science Foundation. This work was done while visiting MADALGO at Aarhus University.}\\
University of Michigan}
\date{}
\end{titlepage}

\maketitle

\begin{abstract}
In this paper, we study the problem of approximating the minimum cut in a distributed message-passing model, the \textsf{CONGEST} model. The minimum cut problem has been well-studied in the context of centralized algorithms. However, there were no known non-trivial algorithms in the distributed model until the recent work of Ghaffari and Kuhn. They gave algorithms for finding cuts of size $O(\epsilon^{-1}\lambda)$ and $(2+\epsilon)\lambda$ in $O(D)+\tilde{O}(n^{1/2+\epsilon})$ rounds and $\tilde{O}(D+\sqrt{n})$ rounds respectively, where $\lambda$ is the size of the minimum cut. This matches the lower bound they provided up to a polylogarithmic factor. Yet, no scheme that achieves $(1+\epsilon)$-approximation ratio is known. We give a distributed algorithm that finds a cut of size $(1+\epsilon)\lambda$ in $\tilde{O}(D+\sqrt{n})$ time, which is optimal up to polylogarithmic factors.
\end{abstract}
\thispagestyle{empty}

\setcounter{page}{1}

\section{Introduction}
The minimum cut problem is a fundamental problem in graph algorithms and network design. Given a weighted undirected graph $G=(V,E)$, a cut $C=(S,V \setminus S)$ where $\emptyset \subset S \subset V$, is a partition of vertices into two non-empty sets. The weight of a cut, $w(C)$, is defined to be the sum of the edge weights crossing $C$. The minimum cut problem is to find a cut with the minimum weight. The exact version of the problem as well as the approximate version have been studied for many years \cite{Karger93, KS93, Karger94b, NI92, Matula93, Gabow95, SW97, Karger2000} in the context of centralized models of computation, resulting in nearly linear time algorithms \cite{Karger2000, Matula93, Karger94b}. 

Elkin \cite{Elkin04} and Das Sarma et al.~\cite{DasSarma12} addressed the problem in the distributed message-passing model. The problem has trivial time complexity of $\Theta(D)$ (unweighted diameter) in the \textsf{LOCAL} model, where the message size is unlimited. Ghaffari and Kuhn \cite{GK13} recently developed approximation algorithms for this problem in the \textsf{CONGEST} model where each message is bounded by $\Theta(\log n)$ bits. They assume that the edges of $G$ have integer weights from $\{1,\ldots, n^{\Theta(1)}\}$ and treat $G$ as an unweighted multigraph, where an edge $e$ with weight $w(e)$ is converted to $w(e)$ parellel edges, while still only $\Theta(\log n)$ bits can be sent over these parallel edges together in each round. Let $\lambda$ be the value of the minimum cut, they give an algorithm that finds a cut of size at most $O(\epsilon^{-1} \lambda)$ in $O(D) + O(n^{1/2+\epsilon} \log^{3} n \log \log n \log^{*} n)$ time. Moreover, they gave an algorithm that finds a cut of size at most $(2+\epsilon)\lambda$ in $O((D+\sqrt{n}\log^{*} n)\log^2 n \log \log n \frac{1}{\epsilon^{5}})$ time. Das Sarma et al.~\cite{DasSarma12} showed $\alpha$-approximating the minimum cut requires $\tilde{\Omega}(D+\sqrt{n})$ rounds for weighted graphs for any $\alpha \geq 1$. Ghaffari and Kuhn extended their lower bound for unweighted multigraphs (which is equivalent to the setting where one is allowed to send messages of size $w \cdot \Theta(\log n)$ over an edge of weight $w$ in weighted graphs). For unweighted simple graphs, they also gave a lower bound of $\tilde{\Omega}(D+\sqrt{n/\alpha})$. Therefore, the upper bound and lower bound provided by Ghaffari and Kuhn match up to a polylogarithmic factor. 

However, still no approximation algorithms exist for any approximation factor less than 2. In this paper, we give a simple algorithm that finds a minimum cut of size at most $(1+\epsilon)\lambda$ in $\tilde{O}(D+\sqrt{n})$ time. In particular, our algorithm runs in $O((\log^{11}n / \epsilon^{17})(D+\sqrt{n}\log^{*} n))$ rounds. 

Our approach uses the semi-duality between minimum cuts and tree packings as in \cite{Karger2000, Thorup07}. Karger \cite{Karger2000} showed that if we greedily pack enough trees, then for any minimum cut, there is a tree crossing the cut at most twice. However, it is technically not easy to utilize this fact to find minimum cuts in the distributed model. Instead, we use a lemma by Thorup \cite{Thorup07}, which shows that if we pack more trees then there is at least one minimum cut that is crossed by a tree exactly once. We take some ingredients from Ghaffari and Kuhn's algorithm and Thurimella's algorithm \cite{Thurimella97} for identifying biconnected components to devise a procedure that is able to simultanously test the values of the $n-1$ cuts induced by deleting one of the $n-1$ edges in a tree. Note that the number of trees we have to pack is polynomial in the value of the minimum cut. Thus, we will first use the sampling lemma of Karger \cite{Karger94} to obtain a sampled graph that scales the value of the minimum cut down to $O(\log n / \epsilon^2)$. Then we only have to pack polylogarithmic number of trees. Finally, we combine the resampling procedure, the tree packing, and the procedure for testing tree-induced-cuts to find an approximate minimum cut.


\section{Distributed Minimum Cut Approximation}
Let $G$ be a connected graph with integer weights from $\{1,\ldots,W \}$, where $W = n^{\Theta(1)}$. We will treat $G$ as a multigraph with uniform edge weights.
Let $\lambda$ be the weight of the minimum cut of $G$. We show how to find such an approximate minimum cut whose weight is at most $(1+\epsilon)\lambda$. 

An edge $e$ is a {\it bridge} if it does not exist a cycle in $G$ passing $e$ (or equivalently, deleting $e$ breaks $G$ into two connected components). Given two graph $A$ and $B$ with the same vertex set, $A+B$ is the multigraph obtained by including edges in $A$ and edges in $B$.

A {\it tree packing} $\mathcal{T}$ is a multiset of spanning trees. The {\it load} of an edge $e$ with respect to $\mathcal{T}$ is the number of trees in $\mathcal{T}$ containing $e$. Given a tree $T$, we say a cut is {\it induced} by $T$ if such a cut is obtained by deleting an edge $e \in T$. We will denote this cut by $C(T,e)$. A tree packing $\mathcal{T} = \{T_1, \ldots, T_k \}$ is {\it greedy} if each $T_i$ is a minimum spanning tree with respect to the loads induced by $\{T_1,\ldots, T_{i-1} \}$. Let $\epsilon' = \Theta(\epsilon)$ such that $(1+\epsilon')^{3} / (1-\epsilon') = 1 + \epsilon$.

\begin{lemma}[Thorup \cite{Thorup07}]\label{lem:treepacking}A greedy tree packing with $96 (\lambda+1)^7 \log^3 m$ trees contains a tree crossing some min-cut only once.
\end{lemma}
\begin{remark}The number of trees in the original statement of the lemma is $\omega(\lambda^7 \log^3 m)$, though the proof actually implies that $\Theta( \lambda^7 \log^3 m)$ is enough. In particular, Thorup showed $24\lambda \ln m /\epsilon^2$ trees is sufficient, where $\epsilon$ satisfies $\frac{\epsilon(3+\log_{1+\alpha} m)}{\lambda} + \alpha < \frac{2}{\lambda(\lambda+1)}$ for some $\alpha < 1$. We can choose $\alpha = \frac{1}{\lambda(\lambda+1)}$ and $\epsilon = \frac{1}{2(\lambda + 1)^3 \ln m}$ to make the inequality hold. Therefore, $96 (\lambda+1)^{7} \ln^3 m$ trees is sufficient.\end{remark}

We describe our algorithm in Algorithm \ref{alg:approxmincut}. The subroutine Test($T, \kappa$) returns a cut whose weight is at most $(1+\epsilon') \kappa$ w.h.p.~if there exists a cut in $G$ induced by $T$ with weight at most $\kappa$.

We show that w.h.p. the algorithm will output a cut $C$ with $w(C) \leq (1+\epsilon)\lambda$. In particular, consider the iteration $i$ where $\lambda \in [X_i, X_{i+1}]$. Let $\lambda'$ denote the value of the minimum cut in the sampled graph $H_i$. If $i = 0$, then it is clear that $\lambda' = \lambda \leq X_1 = 20\ln n / \epsilon'^2$. If $i>0$, since we sampled with probability $1/2^{i} = 20 \ln n / (\epsilon'^2 X_{i+1})= 10\ln n / (\epsilon'^2 X_{i}) \geq 10\ln n / (\epsilon'^2\lambda)$, we know that w.h.p. for any cut $C$ \cite[Corollary 2.4]{Karger94}, $$(1 - \epsilon') \cdot w_{G}(C) / 2^{i} \leq w_{H_{i}}(C) \leq (1 + \epsilon') \cdot w_{G}(C) / 2^{i}.$$
Therefore, $\lambda' \leq (1+\epsilon') \lambda / 2^{i} \leq (1+\epsilon')20\ln n / \epsilon'^2$. If we pack $96 (\lambda'+1)^7 \log^{3} m$ trees in $\mathcal{T}$, then by Lemma \ref{lem:treepacking} there exists a tree crossing some minimum cut $C^{*}$ of $H_i$ only once. Notice that for any other cut $C'$, $$w_{G}(C^{*}) \leq 2^{i} \cdot \frac{w_{H_i}(C^{*})}{1-\epsilon'} = 2^{i} \cdot \frac{\lambda'}{1-\epsilon'} \leq 2^{i} \cdot \frac{w_{H_{i}}(C')}{1-\epsilon'} \leq \frac{1+\epsilon'}{1-\epsilon'} \cdot w_{G}(C')$$

\begin{algorithm}[H]
\caption{$(1+\epsilon)$-approximate minimum cut}
\begin{algorithmic}[1]\label{alg:approxmincut}
\STATE $X_0 \leftarrow 1$
\STATE $i \leftarrow 0$
\REPEAT
\STATE $X_{i+1} \leftarrow 2^{i} \cdot 20\ln n /\epsilon'^2$
\STATE (We are assuming $\lambda \in [X_{i}, X_{i+1}]$ in this iteration)
\STATE Let $H_i$ be the subgraph sampled with probability $p = 1/2^{i}$ on each edge of $G$.
\STATE \label{line:treepacking}Find a greedy tree packing $\mathcal{T}$ with $96((1+\epsilon')20\ln n/\epsilon'^2 + 1)^{7} \ln^{3} m$ trees in $H_i$
\
\STATE $\gamma \leftarrow X_{i}$
\REPEAT \label{line:loopup}
\FOR{each $T \in \mathcal{T}$}
\STATE Call Test($T,(1+\epsilon')\gamma$).
\STATE If Test($T,(1+\epsilon')\gamma$) returns a cut $C$, output $C$ and terminate.
\ENDFOR
\STATE $\gamma \leftarrow (1+\epsilon')\gamma$
\UNTIL{$\gamma > \frac{1+\epsilon'}{1-\epsilon'}\cdot X_{i+1}$} \label{line:loopdown}
\STATE $i \leftarrow i+1$
\UNTIL{$X_{i+1} > nW$}
\end{algorithmic}
\end{algorithm}

Therefore, one of the cuts induced by some $T \in \mathcal{T}$ is an $(1+\epsilon')/(1-\epsilon')$ approximate minimum cut. Denote this cut by $C'$, so $w(C') \in [X_i, ((1+\epsilon')/(1-\epsilon')) \cdot X_{i+1}]$. Therefore in the $i$'th iteration, there exists $\gamma$ in the loop (Line \ref{line:loopup}--Line \ref{line:loopdown}) such that $w(C') \in [\gamma, (1+\epsilon')\gamma]$. So w.h.p.\ we will output a cut with weight at most $(1+\epsilon')^2\gamma \leq (1+\epsilon')^2w(C') \leq (1+\epsilon')^3/(1-\epsilon')w(C^{*}) = (1+\epsilon)w(C^{*})$.

\subsection{Distributed Implmentation}

We have shown the correctness of this algorithm. It remains to show how to implement it in $\tilde{O}(D+\sqrt{n})$ distributed rounds, and in particular, to implement the tree packing (Line \ref{line:treepacking}) and Test($T,\kappa$) in Algorithm \ref{alg:approxmincut}. To pack $k$ trees, it is striaghtfoward to apply $k$ MST computations on the graph where the edge weights are equal to the number of trees including it. This can be done in $O(k(D + \sqrt{n} \log^{*} n))$ rounds \cite{KP95}.

Given a partition $\mathcal{P}$ of $G$ into components, Ghaffari and Kuhn \cite{GK13} devised a testing procedure to test if there is a cut induced by a component in $\mathcal{P}$ that has weight less than $\kappa$ in $\widetilde{O}(D+\sqrt{n})$ rounds. Given a spanning tree $T$, we will show how to test the $n-1$ cuts induced by $T$ also in $\widetilde{O}(D+\sqrt{n})$ rounds. 

\begin{algorithm}[H]
\caption{Test$(T,\kappa)$. Test($T, \kappa$) returns a cut whose weight is at most $(1+\epsilon') \kappa$ w.h.p.~if there exists a cut in $G$ induced by $T$ with weight at most $\kappa$. Note that the sample probability $1-2^{-1/\kappa} = \Theta(1/\kappa)$.}
\begin{algorithmic}[1]\label{alg:testing}
\FOR{$i \leftarrow 1\ldots k = \Theta(\frac{\log n}{\epsilon^2})$}
	\STATE Let $G_i$ be the subgraph obtained by sampling each edge of $G$ independently with probability $1 - 2^{-1/\kappa}$.
	\STATE For each edge $e \in T$, determine if $e$ is a {\bf bridge} in the graph $G_i + T$. \STATE Let $Y_{e,i}=\begin{cases} 1& \mbox{if $e$ is not a bridge in the graph $G_i+T$.} \\ 0 & \mbox{otherwise.}\end{cases} $
\ENDFOR
\STATE If there is $e \in T$ such that $\sum_{i=1}^{k} Y_{e,i} \leq k/2 + \epsilon'k/8$, then return the cut $C(T,e)$
\end{algorithmic}
\end{algorithm}

\begin{lemma} If $T$ induces a cut $C(T,e)$ with weight at most $\kappa$, then Test($T,\kappa$) returns a cut w.h.p. Moreover, any cut returned by the algorithm has weight at most $(1+\epsilon')\kappa$ w.h.p. \end{lemma}

\begin{proof}

Consider a cut $C(T,e)$. First observe that $G_i$ contains an edge crossing $C(T,e)$ if and only if $e$ is not a bridge in the graph $G_i + T$. Therefore, $\E[Y_{e,i}] = 1 - (1 - (1 - 2^{-1/\kappa}))^{w(C(T,e))} = 1 - 2^{-w(C(T,e))/\kappa}$. 

If there is $C(T,e) \leq \kappa$, then $\E[Y_{e,i}] \leq 1/2$ and $\E[\sum_i Y_{e,i}] \leq k/2$. By Hoeffiding's inequality, $\Pr(\sum_i Y_{e,i} > k/2 + \epsilon' k/ 8) \leq \Pr(\sum_i Y_{e,i} > \E[\sum_i Y_{e,i}] + \epsilon' k/ 8) \leq e^{-\frac{2(\epsilon'k/8)^2}{k}} = e^{-\epsilon'^2 k/32} = 1/\poly(n)$. By taking the union bound over the $n-1$ cuts induced by $T$, we conclude that w.h.p. the algorithm will return a cut if there is cut whose weight is at most $\kappa$.

On the other hand if $w(C(T,e)) > (1+\epsilon')\kappa$, then $\E[Y_{e,i}] = 1 - 2^{-{1-\epsilon'}} \geq 1/2 + \epsilon' / 4$ when $\epsilon' \leq 1$, since $2^{-\epsilon'} \leq 1 - \epsilon'/2$ when $\epsilon' \leq 1$. So $\E[\sum_i Y_{e,i}] \geq k/2 + \epsilon'k/ 4$. By Hoeffiding's inequality, $\Pr(\sum_i Y_{e,i} \leq k/2 + \epsilon' k/ 8) \leq \Pr(\sum_i Y_{e,i} \leq \E[\sum_i Y_{e,i}] - \epsilon' k / 8) \leq e^{-\frac{2(\epsilon'k/8)^2}{k}} = e^{-\epsilon'^2 k/32} = 1/\poly(n)$. By taking the union bound over the $n-1$ cuts induced by $T$, we conclude the cut returned by the algorithm has weight at most $(1+\epsilon')\kappa$ w.h.p.
\end{proof}

\subsection{Computing the Bridges}
Given a subgraph $G_i$ of $G$, it remains to show how to determine what edges of $T$ are bridges in the subgraph $T +G_i$ in $\tilde{O}(D+\sqrt{n})$ rounds. Thurimella \cite{Thurimella97} gave an algorithm for computing the biconnected components of the underlying graph in $\tilde{O}(D+\sqrt{n})$ rounds. With simple modifications, it can be applied to compute which edges of $T$ are bridges in the {\it subgraph} $G_i$ of the underlying graph $G$. Note that even we have the algorithm for computing the bridges of $T$ in $G+T$, it is not clearly whether we can directly simulate it to compute the bridges of $T$ in $G_i+T$, because we want the running time to depend on the diameter of $G$ rather than that of $G_i$. Therefore, we describe the algorithm and necessary changes below.

\renewcommand\thefootnote{\fnsymbol{footnote}}
Fix a root $r$ in $T$. Let $\mathit{pre}(u) \in [0,n-1]$ be the preorder number which denote the time $u$ is visited if we perform a depth-first search on $T$ starting at $r$. Denote the subtree rooted at $u$ by $T_u$ and let $\mathit{size}(u)$ be the size of $T_u$. Let
\addtocounter{footnote}{1} 
\begin{align*}
\mathit{low}(u) &\defeq \min \begin{cases}\mathit{pre}(u) \\ \mathit{low}(v) & \mbox{$v$ is a child of $u$ in $T$} \\ \mathit{pre}(v) & \mbox{$uv \in G_i $ \footnotemark}\end{cases}
\\ \mathit{high}(u) &\defeq \max \begin{cases}\mathit{pre}(u) \\ \mathit{high}(v) & \mbox{$v$ is a child of $u$ in $T$} \\ \mathit{pre}(v) & \mbox{$uv \in G_i$ \addtocounter{footnote}{-1}\footnotemark }\end{cases}
\end{align*}
 \footnotetext[2]{It can be the case that $v$ is the parent of $u$ in $T$, which happens when there are parallel edges between $u$ and $v$ in $G_i + T$, and one of them is in $T$. Note that an edge is not a bridge if it is a multiedge.}
\begin{lemma}Let $uv \in T$ with $v$ being a child of $u$, i.e.~$\mathit{pre}(u) < \mathit{pre}(v)$. $uv$ is a bridge if and only if $\mathit{low}(v) \geq \mathit{pre}(v)$ and $\mathit{high}(v) \leq \mathit{pre}(v) + \mathit{size}(v) - 1$.\end{lemma}
\begin{proof}
First notice that every vertex $x \in T_v$ must have $\mathit{pre}(x) \in [\mathit{pre}(v), \mathit{pre}(v) + \mathit{size}(v) - 1]$. If $uv$ is a bridge, then no descendent of $v$ will be adjacent to anything outside the subtree rooted at $v$, for otherwise a cycle passing $uv$ will be created. Therefore, $\mathit{low}(v), \mathit{high}(v) \in [\mathit{pre}(v), \mathit{pre}(v) + \mathit{size}(v) - 1]$.

On the other hand, if $\mathit{low}(v) < \mathit{pre}(v)$ or if $\mathit{high}(v) \geq \mathit{pre}(v) + \mathit{size}(v)$, then there exists a vertex $y \in T_v$ and $z \notin T_v$ such that $y$ and $z$ are adjacent. Since $z \notin T_v$, there must exists a path from $z$ to $u$ such that it does not pass $uv$. Therefore, $u \to v \leadsto y \to z \leadsto u$ forms a cycle and $uv$ is not a bridge.
\end{proof}
\begin{remark}Note that the second condition $\mathit{high}(v) \leq \mathit{pre}(v) + \mathit{size}(v) - 1$ is needed because $T$ is not necessarily a DFS tree. \end{remark}

Now it remains to show how to compute $\mathit{pre}(u)$, $\mathit{low}(u)$, and $\mathit{high}(u)$ in $\tilde{O}(D+\sqrt{n})$ time. It is explicitly described in \cite{Thurimella97} how to compute $\mathit{pre}(u)$. Note that $\mathit{pre}(u)$ is independent of the $G_i$. Although $\mathit{low}(u)$ and $\mathit{high}(u)$ depend on $G_i$, they can be computed in a similar way in $\tilde{O}(D+\sqrt{n})$ time. For completeness, we describe how to compute these functions in the following.

\begin{lemma}[\cite{GKP93,KP95}]\label{lem:decompose} A tree of $n$ vertices can be divided into $O(\sqrt{n})$ connected subgraphs each of diameter $O(\sqrt{n})$ in $O(\sqrt{n}\log^{*}n)$ time\end{lemma}

First, use Lemma \ref{lem:decompose} to decompose the rooted tree $T$ into components $F_1,\ldots F_{O(\sqrt{n})}$. For each component $F_i$, there is a root $r_i$ which is either the root of $T$, $r$, or the unique vertex in $F_i$ connecting to its parent outside $F_i$. It is shown in \cite{Peleg} that the root $r$ is able to downcast distinct messages of size $O(\log n)$ to each of $r_i$ in $O(D+\sqrt{n})$ time. Conversely, it is possible for each of the $r_i$ to upcast a message of size $O(\log n)$ to the root $r$ in $O(D+\sqrt{n})$ time. 

Suppose each vertex has a unique ID. The component ID of $F_i$ is defined to be the ID of $r_i$. The component ID can be broadcast to every vertex in the component in $O(\sqrt{n})$ rounds. We can then assume that the root $r$ knows the topology of the contracted tree where each component is contracted into a single vertex. This can be done if every root $r_i$ upcasts a message about the component ID of its parent and itself.

To compute $\mathit{pre}(u)$, each root $r_i$ in each component first calculate the size of $F_i$ then upcast it to $r$. Since $r$ knows the topology of the contracted tree, $r$ can calculate the size of each subtree rooted at each of $r_i$. Then $r$ downcasts the size of subtree rooted at $r_i$ back to $r_i$. Now each $F_i$ computes its preorder number internally in $O(\sqrt{n})$ time assuming $r_i$ has number $0$. During the computation, each $r_i$ also records what its preorder number is supposed to be if the depth-first search started from the root of its parent component. Finally, each $r_i$ upcasts this number to $r$ and then $r$ computes the correct offset for each subtree and downcasts the offsets back to the $r_i$. After adding the offset internally, we get the correct preorder number.

To compute $\mathit{low}(u)$, initially each vertex $u$ computes $\min(\mathit{pre}(u),\min_{uv \in G_i}\mathit{pre}(v))$ in constant rounds. Then the problem becomes aggregating the minimum in the subtree $T_u$ for each $u$. First, each $r_i$ computes the minimum in $F_i$ in $O(\sqrt{n})$ time and then upcasts to $r$. Using the information, $r$ calculates the minimum of the subtrees rooted at each $r_i$ and downcasts to each $r_i$. Now each $r_i$ sends the minimum to its parent via the inter-component links. The parent replace its minimum if it is smaller. Finally, each component $F_i$ internally updates the minimum toward the root $r_i$. Then each vertex has the correct minimum. $\mathit{high}(u)$ can be computed in the same way.

Therefore, the step of computing the bridges in $T$ of $G_i + T$ takes $O(D + \sqrt{n} \log^{*} n)$ time. Each invocation of Test($T,\kappa$) takes $O(\frac{\log n}{\epsilon^2}(D+\sqrt{n} \log^{*} n))$ time. 

\subsection{Running Time}
Now we analyze the running time of Algorithm \ref{alg:approxmincut}. The outerloop runs for $O(\log n)$ iterations. Therefore, the tree packing, Line \ref{line:treepacking}, is executed $O(\log n)$ times, each taking $O(\log^{10} n / \epsilon^{14}(D + \sqrt{n}\log^{*} n))$ rounds.

Let $k = O(\log (nW))$ be the largest index such that $X_k \leq nW$. The total number of iterations that the innerloop runs is at most $$\sum_{i=0}^{k} \log_{1+\epsilon'} \left(\frac{1+\epsilon'}{1-\epsilon'} \cdot \frac{X_{i+1}}{X_{i}}\right)
= O(k) + \sum_{i=0}^{k} \log_{1+\epsilon'}\frac{X_{i+1}}{X_i} = O(k) + \log_{1+\epsilon'} (X_{k+1}) = O(\log n/ \epsilon)$$
Therefore, Test($T,\kappa$) is invoked at most $O( (\log n / \epsilon)\cdot (\log^{10} n / \epsilon^{14}) )$ times, each taking $O((\log n/ \epsilon^2)(D + \sqrt{n} \log^{*} n))$ rounds. 

The total running time is
\begin{align*}&O(\log n \cdot (\log^{10} n / \epsilon^{14})(D + \sqrt{n}\log^{*} n) ) + (\log^{11} n / \epsilon^{15})\cdot((\log n/ \epsilon^2)(D + \sqrt{n} \log^{*} n)) \\ &= O((\log^{12} n / \epsilon^{17}) \cdot (D + \sqrt{n}\log^{*} n) ) = \tilde{O}(D+\sqrt{n}) \end{align*}

\begin{remark}\label{rmk:combinedGK} The total iterations of the outerloop and innerloop in Algorithm \ref{alg:approxmincut} can be reduced to $O(1)$ and $O(1/ \epsilon)$ by first approximating $\lambda$ within constant factor by Ghaffari and Kuhn's algorithm. Then, we can reduce our running time to $O((\log^{11}n / \epsilon^{17})(D+\sqrt{n}\log^{*} n))$.\end{remark}

The exponent of the $\log n$ and the $\epsilon$ in our running time depends heavily on the size of the greedy tree packing in Lemma \ref{lem:treepacking}. If one can show that $O(\lambda^{a} \log^b n)$ trees is sufficient, then our running time can be improved to $O((\log^{2+a+b} n / \epsilon^{2a+3})\cdot(D+ \sqrt{n} \log^{*} n) )$ rounds. Using Ghaffari and Kuhn's algorithm to approximate $\lambda$ within a constant (Remark \ref{rmk:combinedGK}), we can get a running time of $O((\log^{1+a+b} n / \epsilon^{2a+3} + (\log^2 n \log \log n)/\epsilon^5)\cdot(D+ \sqrt{n} \log^{*} n) )$. For comparison, Karger \cite{Karger2000} showed that a greedy tree packing of size $O(\lambda \log n)$ is enough for any minimum cut to be crossed at most twice by some tree. It will be interesting to see if the number of trees in Lemma \ref{lem:treepacking} can be reduced.

\bibliography{mybib_short}

\begin{thebibliography}{10}

\bibitem{DasSarma12}
A.~Das~Sarma, S.~Holzer, L.~Kor, A.~Korman, D.~Nanongkai, G.~Pandurangan,
  D.~Peleg, and R.~Wattenhofer.
\newblock Distributed verification and hardness of distributed approximation.
\newblock {\em SIAM Journal on Computing}, 41(5):1235--1265, 2012.

\bibitem{Elkin04}
M~Elkin.
\newblock Distributed approximation: A survey.
\newblock {\em SIGACT News}, 35(4):40--57, 2004.

\bibitem{Gabow95}
H.~N. Gabow.
\newblock A matroid approach to finding edge connectivity and packing
  arborescences.
\newblock {\em Journal of Computer and System Sciences}, 50(2):259 -- 273,
  1995.

\bibitem{GKP93}
J.~A. Garay, S.~Kutten, and D.~Peleg.
\newblock A sub-linear time distributed algorithm for minimum-weight spanning
  trees.
\newblock In {\em Proc. 34th Symposium on Foundations of Computer Science
  (FOCS)}, pages 659--668, 1993.

\bibitem{GK13}
M.~Ghaffari and F.~Kuhn.
\newblock Distributed minimum cut approximation.
\newblock In {\em Proc. 27th Symposium on Distributed Computing (DISC)}, volume
  8205, pages 1--15. 2013.

\bibitem{Karger93}
D.~R. Karger.
\newblock Global min-cuts in {RNC}, and other ramifications of a simple min-out
  algorithm.
\newblock In {\em Proc. 4th Annual ACM-SIAM Symposium on Discrete Algorithms
  (SODA)}, pages 21--30, 1993.

\bibitem{Karger94}
D.~R. Karger.
\newblock Random sampling in cut, flow, and network design problems.
\newblock In {\em Proce. 26th ACM Symposium on Theory of Computing (STOC)},
  pages 648--657, 1994.

\bibitem{Karger94b}
D.~R. Karger.
\newblock Using randomized sparsification to approximate minimum cuts.
\newblock In {\em Proc. 5th Annual ACM-SIAM Symposium on Discrete Algorithms
  (SODA)}, pages 424--432, 1994.

\bibitem{Karger2000}
D.~R. Karger.
\newblock Minimum cuts in near-linear time.
\newblock {\em J. ACM}, 47(1):46--76, January 2000.

\bibitem{KS93}
D.~R. Karger and C.~Stein.
\newblock An $\tilde{O}(n^2)$ algorithm for minimum cuts.
\newblock In {\em Proc. 25th ACM Symposium on Theory of Computing (STOC)},
  pages 757--765, 1993.

\bibitem{KP95}
S.~Kutten and D.~Peleg.
\newblock Fast distributed construction of k-dominating sets and applications.
\newblock In {\em Proc. 14th ACM Symposium on Principles of Distributed
  Computing (PODC)}, pages 238--251, 1995.

\bibitem{Matula93}
D.~W. Matula.
\newblock A linear time $2+\epsilon$ approximation algorithm for edge
  connectivity.
\newblock In {\em Proc. 4th Annual ACM-SIAM Symposium on Discrete Algorithms
  (SODA)}, pages 500--504, 1993.

\bibitem{NI92}
H.~Nagamochi and T.~Ibaraki.
\newblock Computing edge-connectivity in multigraphs and capacitated graphs.
\newblock {\em SIAM J. Discret. Math.}, 5(1):54--66, February 1992.

\bibitem{Peleg}
D.~Peleg.
\newblock {\em Distributed Computing: A Locality-Sensitive Approach}.
\newblock Monographs on Discrete Mathematics and Applications. Society for
  Industrial and Applied Mathematics, 2000.

\bibitem{SW97}
M.~Stoer and F.~Wagner.
\newblock A simple min-cut algorithm.
\newblock {\em J. ACM}, 44(4):585--591, July 1997.

\bibitem{Thorup07}
M.~Thorup.
\newblock Fully-dynamic min-cut.
\newblock {\em Combinatorica}, 27(1):91--127, 2007.

\bibitem{Thurimella97}
R.~Thurimella.
\newblock Sub-linear distributed algorithms for sparse certificates and
  biconnected components.
\newblock {\em Journal of Algorithms}, 23(1):160 -- 179, 1997.

\end{thebibliography}
\bibliographystyle{plain}

\end{document}